\documentclass[12pt,colorinlistoftodos]{article}
\usepackage{amssymb, amsmath, amsbsy, amsfonts}
\usepackage{amsfonts}
\usepackage{booktabs}
\usepackage[T1]{fontenc}
\usepackage{amsthm}
\usepackage{amscd}
\usepackage{graphicx}
\usepackage{latexsym}
\usepackage{xcolor}
\usepackage[colorlinks]{hyperref}
\usepackage{subeqnarray}
\usepackage{eqnarray}
\usepackage[utf8]{inputenc}
\usepackage{amsmath}
\usepackage{amssymb}
\usepackage{amsthm}
\usepackage{epsfig}
\usepackage{hyperref}
\usepackage{subcaption}
\usepackage{soul}
\usepackage{float}
\usepackage{todonotes}

\newcommand{\M}{\mathcal{M}}
\newcommand{\R}{\mathcal{R}}

\newcommand{\V}{\mathcal{V}}

\newcommand{\IR}{\mathbb{R}}

\usepackage{bm}
\def\bb{\bm{b}}
\def\bd{\bm{d}}
\def\bv{\bm{v}}
\def\bx{\bm{x}}
\def\by{\bm{y}}

\def\bE{\bm{E}}

\def\bI{\bm{I}}

\def\bN{\bm{N}}
\def\bR{\bm{R}}
\def\bS{\bm{S}}
\def\bW{\bm{W}}
\def\bX{\bm{X}}
\def\bZ{\bm{Z}}

\def\bvarepsilon{\bm{\varepsilon}}
\def\bbeta{\bm{\beta}}
\def\bgamma{\bm{\gamma}}
\def\blambda{\bm{\lambda}}

\def\bzeta{\bm{\zeta}}
\def\bdelta{\bm{\delta}}

\def\b0{\bm{0}}
\def\bOmega{\bm{\Omega}}
\def\diag{\mathsf{diag}}

\newtheorem{theorem}{Theorem}

\newtheorem{lemma}[theorem]{Lemma}
\newtheorem{proposition}[theorem]{Proposition}

\usepackage[normalem]{ ulem }
\usepackage{soul}
\usepackage{cancel}

\usepackage{tikz}
\usetikzlibrary{shapes.geometric, arrows}
\usepackage{pgfplots}
\pgfplotsset{compat=1.15}
\usepackage{mathrsfs}
\usetikzlibrary{arrows}
\tikzstyle{startstop} = [rectangle, rounded corners, minimum width=3cm,minimum height=1cm,text centered, draw=black, fill=blue!30]
\tikzstyle{io} = [trapezium, trapezium left angle=70, trapezium right angle=110, minimum width=3cm, minimum height=1cm, text centered, draw=black, fill=blue!30]
\tikzstyle{process} = [rectangle, trapezium left angle=70, trapezium right angle=110, minimum width=3cm, minimum height=1cm, text centered, draw=black, fill=green!30]
\tikzstyle{arrow} = [thick,->,>=stealth]
\PassOptionsToPackage{force}{filehook}
\usepackage{tikz}
\usetikzlibrary{shapes,arrows}
\usetikzlibrary{positioning}
\tikzstyle{cloud} = [draw, ellipse,fill=red!20, node distance=0.87cm,
minimum height=2em]
\tikzstyle{line} = [draw, -latex']
\usetikzlibrary{shapes.symbols,shapes.callouts,patterns}
\usetikzlibrary{calc}
 \numberwithin{dummy}{section}
 
\providecommand{\keywords}[1]
{
	\small	
	\textbf{\textit{Keywords---}} #1
}

\usepackage{amsmath}
\usepackage{todonotes}
\usepackage{booktabs}
\usepackage{enumitem}
\usepackage{geometry}
\title{A metapopulation model for the spread of cholera}
\author{Abdramane Annour Saad$^{\bf a}$ \and Julien Arino$^{\bf b}$ \and Patrick M Tchepmo Djomegni$^{\bf c}$\and Mahamat S Daoussa Haggar$^{\bf a}$ \\
\small $^{\bf a}$Laboratory L2MIAS, University of N'Djamena, Chad\\
\small	$^{\bf b}$Department of Mathematics, University of Manitoba, Winnipeg, MB, Canada\\
\small	$^{\bf c}$School of Mathematical and Statistical Sciences, North-West University, South Africa}

\begin{document}

\maketitle
\begin{abstract}
    We consider a metapopulation model of cholera describing explicitly the movement of individuals and contaminated water between locations as well as a simple vaccination mechanism.
    The global stability of the disease-free equilibrium point when the location-specific reproduction numbers are all less than unity is established.
    We then conduct some numerical investigation of the model.
\end{abstract}

\section{Introduction}
Cholera is an intestinal disease caused when the bacterium \emph{Vibrio cholerae} colonises the human intestine \cite{colwell1996global}. 
Mechanisms of cholera transmission include water, unwashed contaminated food and seafood \cite{griffith2006review}. 
The symptoms of cholera are severe and profuse diarrhoea leading to severe dehydration. 
In the absence of prompt treatment, its classical form can result in death in over half of the cases \cite{eliot2012interpreter}.
The incubation period ranges from hours to 5 days; shortly thereafter, symptoms begin abruptly. Clinical forms include dry cholera, in which the infected individual enters a state of shock and dies rapidly; classic form, characterised by diarrhoea and vomiting and can lead to death; and mild form, in which there is no apparent symptomatology and the disease is only diagnosed through laboratory tests. 
Treatment involves rehydration and the use of antibiotics, while prevention is achieved through sanitation measures, water treatment and care with personal hygiene and nutrition \cite{longini2007controlling}.

The motivation for this paper is the situation around Lake Chad and more generally, in the Lake Chad Basin.
Since the 1990s, the burden of cholera incidence has shifted from the Americas to Africa, with over 99\% of reported deaths of cholera in 2009 taking place there \cite{WHO2010cholera}, for instance.
The Lake Chad Basin itself has become one of the hotspots of cholera on the continent and a region where cholera is both endemic and subject to regular outbreaks \cite{koua2025exploring}.

Lake Chad is an endorheic freshwater lake.
As a consequence of desertification and diversion of water for human use, Lake Chad's surface area has dramatically dropped and now varies between 2,000 and 5,000 square kilometres depending on the season and the year \cite{coe2001human}, extending even further some years during floods.
Even with this reduced size, Lake Chad borders four countries: Cameroon, Nigeria and the landlocked Chad and Niger.
Because Lake Chad is endorheic, contaminants, including cholera, accumulate or persist longer than they would with exorheic lakes or basins.
The catchment area of Lake Chad, the so-called Lake Chad Basin, is large, covering, depending on the precise definition used, from 1 to 2.5 million square kilometres \cite{magrin2015atlas} and comprising areas of the four countries neighbouring the lake as well as areas in Algeria, the Central African Republic, Libya and Sudan \cite{asah2015transboundary}.
Altogether, Lake Chad is home to 2--3 million people and the wider Lake Chad Basin to dozens of millions people.
These people hail from varied ethnic groups \cite{magrin2015atlas} and are further subject to the presence of country boundaries that are, in some instances, not necessarily observed by, for instance, nomadic herders living in the area.
The situation is further complicated by complex interactions between ethnic groups, countries and a degrading environment \cite{owonikoko2020environmental}, as well as terrorism in the area \cite{magrin2018crisis}.
As a water-borne infection, the spatial and temporal patterns of cholera epidemics are closely linked to water flows and the ecology of the bacterium in the environment, which is itself influenced by weather conditions and climate variability.
Furthermore, the movement of population in the area is equally important.


Regarding its spatio-temporal spread, cholera is the first disease to have been studied under the lens of geographical medicine, in the 1854 seminal work of Snow \cite{mcleod2000our,tulchinsky2018john}.
However, in the mathematical epidemiology sphere, spatial aspects of cholera have not been studied much.
In a recent (2024) systematic review of modelling techniques in cholera epidemiology, \cite{anteneh2024modelling} identified 476 relevant papers; of these, only 10\% considered spatial aspects and most of them used statistical techniques.
Very few models were formulated using partial differential equations models; see, e.g., \cite{nnaji2024spatio}.
Agent-based models (ABM), on the other hand, were used more because they allow for a level of realism that other model types fail to achieve; see, e.g., the models in \cite{augustijn2016agent,crooks2014agent}.

Our focus here is on another type of models, so-called metapopulation models, which expand classic compartmental epidemiological models by incorporating space as discrete locations between which movement is possible \cite{Arino2017}.
We are aware of very few models of this type specific to cholera.
In \cite{grad2012cholera}, an epidemic model for cholera is set in a metapopulation framework; because it is epidemic, this model focuses on outbreaks.
In \cite{bertuzzo2016probability}, a model is considered that does not implicitly incorporate any sort of movement except for water.
Other models \cite{bertuzzo2008space,bertuzzo2010spatially,mari2012modelling} by the same team also only consider movement of water and implicit long distance coupling for humans.
The model in \cite{trevisin2022epidemicity} is extremely detailed in each patch.
It is also, like the model in \cite{bertuzzo2016probability}, not explicit for movement except for the water component.
A few models consider SIR models with a bacterial count compartment but only in two locations; see, e.g., \cite{amadi2022metapopulation,njagarah2014metapopulation}.
Our focus here is the model in \cite{tuite2011cholera}, which expands a well-known basic model of \cite{codecco2001endemic} to the population of each of the 10 administrative departments of Haiti. 
Movement of humans is included implicitly, in the sense that the term describing infection in a given location includes dependence on the number of infectious humans in other locations.
Save for the different nature of the incidence function used there and the presence of vaccination, the resulting model is in essence similar to the model \cite{senapati2019cholera}.    

The objective of our work is to extend the model in \cite{tuite2011cholera}, taking into account explicitly the movement of both humans and contaminated water.
This will have the effect of bringing that model more in line with the model in \cite{senapati2019cholera}.
We ensure that the model is sound by performing some mathematical analysis of its properties.
We then proceed to a computational analysis. Of particular interest to use is whether it is necessary to include water movement in simple models, both implicitly as \cite{bertuzzo2008space,bertuzzo2010spatially,mari2012modelling} or explicitly as \cite{senapati2019cholera}.
This paper is organised as follows. 
In Section~\ref{sec:model}, we present the metapopulation model, while the analytical study of the model is presented in Section~\ref{sec:analysis}. 
We conclude with numerical considerations in Section~\ref{sec:computational}.

\section{The metapopulation model for cholera}
\label{sec:model}
\subsection{Formulation of the model}
The model is a variation on the model of \cite{tuite2011cholera}. 
We assume that there is a set of $n$ locations (also called \emph{patches}) between which cholera can spread when either an individual bearing the pathogen or contaminated water moves from one location to another.
In the model of \cite{tuite2011cholera}, contaminated water movement is \emph{implicit}, whereas we make it \emph{explicit} here.
Doing so allows to incorporate information about water flows, if available.

In a given location $i=1,\ldots,n$, we denote $N_i(t)$ the total human population and $W_i(t)$ the concentration of bacteria in the water at time $t\geq 0$.
The human population is further divided into three distinct compartments, with $S_i(t)$, $I_i(t)$ and $R_i(t)$ denoting the number of individuals susceptible to the disease, $I_i(t)$infected by it, $I_i(t)$ and recovered from it, $R_i(t)$, so that $N_i(t)=S_i(t)+I_i(t)+R_i(t)$.
In the sequel, time dependence of the state variables is omitted if this does not lead to confusion.

Susceptible individuals become infected either by contact with infected individuals, with contact parameter $\beta_i^I$, or by contact with contaminated water, with contact parameter $\beta_i^W$.
Contrary to \cite{tuite2011cholera}, it is assumed here that transmission is \emph{local}, in the sense that individuals only acquire infection from water or individuals present in the location in which they are currently located.
Note also that contrary to the models in, say, \cite{bertuzzo2008space,bertuzzo2010spatially,mari2012modelling,senapati2019cholera}, we do not assume that the force of infection of bacteria acting on humans is a saturating function of the bacterial load $W_i$ and instead use the linear term $\beta_i^WW_i$, giving mass action incidence for both human-to-human and water-to-human transmission.

We assume that individuals can be vaccinated against the disease and denote $v_i$ the \emph{per capita} vaccination rate of susceptible individuals. 
In this simple model, we assume that vaccination moves susceptible individuals into the recovered compartment, affording them the temporary immunity enjoyed by other individuals there.
While not as realistic as models with an explicit vaccination compartment, this avoids the complication of potential so-called backward bifurcations, which complicate the mathematical analysis without necessarily adding much information from an epidemiological perspective \cite{annoursaad2025modelcholerainfectiousnessdeceased}.

Infected individuals can in turn contaminate the water compartment by shedding the pathogen at a \emph{per capita} rate $\zeta_i$.
Note that this is not a flow between compartments in the usual sense: shedding is a consequence of infection but does not decrease the number of infected humans.

Regarding demographic processes, we assume that humans are recruited into the susceptible compartment (there is no vertical transmission of cholera) at the constant rate $b_i$ and are subject to natural death at the \emph{per capita} rate $d_i^H$ regardless of their disease status.
\emph{V. cholerae} does not survive indefinitely in the water; it dies at the \emph{per capita} rate $d_i^C$.

Concerning movement between patches, we assume the following.
Humans move from location $j$ to location $i$ at a rate $m_{ij}$, which can depend on disease status, which we denote $m_{ij}^X$ for $X\in\{S,I,R\}$.
There is no death during movement.

Contaminated water can also move between locations, because of river flows, flooding or lake water movement.
We denote $m_{ij}^W$ the rate of movement of contaminated water from location $j$ to location $i$. 
Note that although mortality of vibrio is not explicitly incorporated in $m^W$, this is easily done. 
Suppose for instance that $\tilde m_{ij}^W$ is a known uncontaminated water movement rate from patch $j$ to patch $i$. If this transport takes $t_{ij}^W$ time units, then using $m_{ij}^W=\exp(-d^C_jt_{ij}^W)\tilde m_{ij}^W$ incorporates survival of vibrio during transport from $j$ to $i$ if the lifetime of individual bacteria is exponentially distributed.


\begin{figure}[htbp]
    \centering
    \begin{tikzpicture}[scale=0.8, 
			every node/.style={transform shape},
			auto,
			box/.style={minimum width={width("N-1")+12pt},
				draw, rectangle}]
        \node [box, fill=green!30] at (0,4) (SJ) {$S_j$};
        \node [box, fill=green!30] at (0,10.5) (SI) {$S_i$};
        \node [box, fill=blue!20] at (2.5,1.8) (WXJ) {$W_j$};
        \node [box, fill=red!40] at (5.9,4) (IXJ) {$I_j$};
        \node [box, fill=yellow!20] at (8.9,1.7) (RXJ) {$R_j$};
        \node [box, fill=blue!20] at (2.5,7.8) (WXI) {$W_{i}$};
        \node [box, fill=red!40] at (6,10.5) (IXI) {$I_{i}$};
        \node [box, fill=yellow!20] at (8.9,7.8) (RXI) {$R_{i}$};
        \node[rotate=90] at (-2,2.9) {\Large\textbf{Patch j}};
        \node[rotate=90] at (-2,9.25) {\Large\textbf{Patch i}};
        \draw [black,rounded corners] (-1.65,-1) -- (-1.65,5.8) -- (11.4,5.8) -- (11.4,-1) -- cycle;
        \draw [black,rounded corners] (-1.65,6.3) -- (-1.65,13) -- (11.4,13) -- (11.4,6.3) -- cycle;
        \path [line,  thick] (SJ) to node [above, sloped] (TextNode) {$\quad\quad\quad\quad\quad\quad\quad\lambda_j$} (IXJ);
        \path [line,  thick] (SJ) to node [pos=0.7,above, sloped] (TextNode) {$\upsilon_j S_j$} (RXJ);
        \path [line,  thick] (IXJ) to node [above,sloped] (TextNode) {$\gamma_{j} I_j$} (RXJ);
        \path [line,  thick] (IXI) to node [above,sloped] (TextNode) {$\gamma_{i} I_i$} (RXI); 
        \path [line, thick, dashed, bend left] (IXJ) to node [midway, below, sloped] (TextNode) {$\zeta_{j} I_j$} (WXJ);
        \path [line,   bend left=10,thick] (SI) to node [above, sloped] (TextNode) {$m_j S_i$} (SJ);
        \path [line, bend left=10, thick] (SJ) to node [above, sloped] (TextNode) {$m_i S_j$} (SI);
        \path [line,  bend left=10, thick] (WXJ) to node [above, sloped] (TextNode) {$m_j W_i$} (WXI);
        \path [line, bend left=10, thick] (WXI) to node [above, sloped] (TextNode) {$m_i W_j$} (WXJ);
        \path [line,  bend left=10, thick] (IXI) to node [above, sloped] (TextNode) {$m_j I_i$} (IXJ);
        \path [line,bend left=10, thick] (IXJ) to node [above, sloped] (TextNode) {$m_i I_j$} (IXI);
        \path [line,  bend left=10, thick] (RXI) to node [above, sloped] (TextNode) {$m_j R_i$} (RXJ);
        \path [line, bend left=10, thick] (RXJ) to node [pos=0.6,above, sloped] (TextNode) {$m_i R_j$} (RXI);
        \path [line, thick] (SI) to node [midway, above, sloped] (TextNode) {$\lambda_i S_i$} (IXI);
        \path [line, thick, dashed, bend left] (IXI) to node [midway, below, sloped] (TextNode) {$\zeta_{i} I_i$} (WXI);
        
        \path [line, thick] (SI) to node [pos=0.8, above, sloped] (TextNode) {$\upsilon_i S_{i}$} (RXI);
        \path [line, thick, bend left=9.5cm] (RXI) to node [pos=0.53, above, sloped] (TextNode) {$\varepsilon_i R_{i}$} (SI);
        \path [line, thick, bend right=9.5cm] (RXJ) to node [pos=0.5, above, sloped] (TextNode) {$\varepsilon_j R_{j}$} (SJ);
        \coordinate[left=1cm of SI] (r1);
        \coordinate[left=1cm of SJ] (r2);
        \coordinate[left=0.95cm of WXI] (w1);
        \coordinate[below of=WXJ] (w2);
        \coordinate[right=1.6cm of IXI] (di1);
        \coordinate[right=1.6cm of IXJ] (di2);
        \coordinate[right=1.2cm of RXI] (dr1);
        \coordinate[below of=RXJ] (dr2);
        \coordinate[above=1.4cm of SI] (ds1);
        \coordinate[below of=SJ] (ds2);
        \path [line, thick] (r1) to node [midway, above] (TextNode) {$b_i$} (SI);
        \path [line,  thick] (r2) to node [midway, above] (TextNode) {$b_j$} (SJ);
        \path [line,  thick] (WXI) to node [midway, above] (TextNode) {$d^C_i W_i$} (w1);
        \path [line,  thick] (WXJ) to node [midway, left] (TextNode) {$d^C_jW_j$} (w2);
        \path [line, thick] (IXI) to node [midway, below] (TextNode) {\tiny$ \delta_iI_i+d^H_iI_i$} (di1);
        \path [line,  thick] (IXJ) to node [midway, above] (TextNode) {\tiny$\delta_jI_j + d^H_jI_j$} (di2);
        \path [line,  thick] (RXI) to node [midway, above] (TextNode) {$ d^H_iR_i$} (dr1);
        \path [line,  thick] (RXJ) to node [midway, right] (TextNode) {$d^H_jR_j$} (dr2);
        \path [line,  thick] (SI) to node [midway, left] (TextNode) {$d^H_iS_i$} (ds1);
        \path [line,  thick] (SJ) to node [midway, right] (TextNode) {$d^H_j S_j$} (ds2);
    \end{tikzpicture}
    \caption{Flow diagram of the SIRW model in the case of two patches.}
    \label{fig:flow_diagram}
\end{figure}
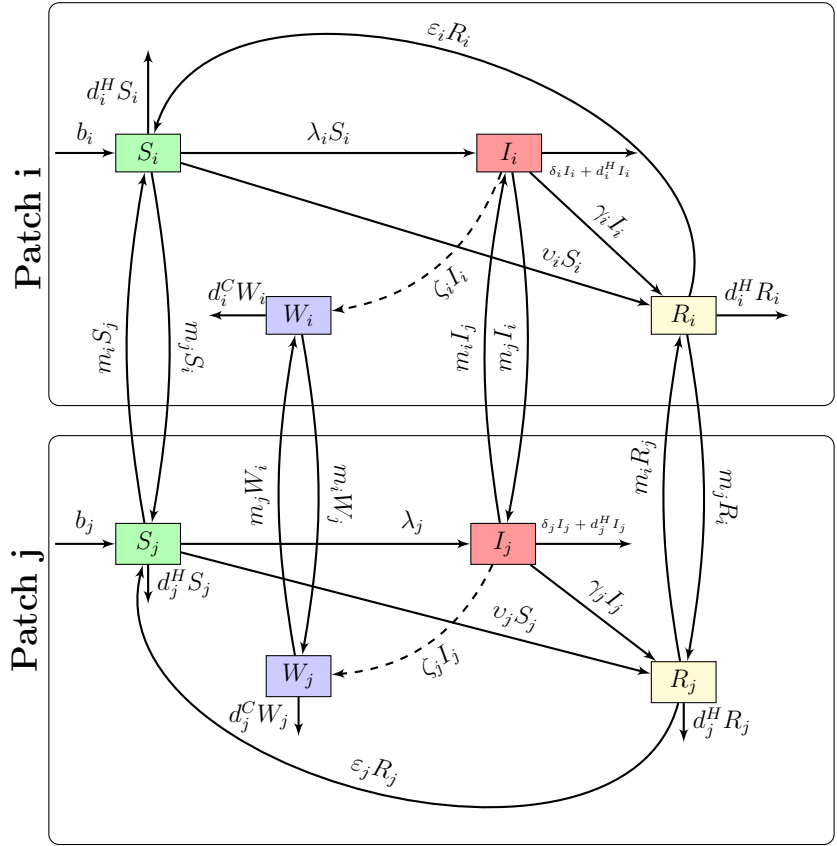

Let $X_i$ be the population or density of compartment $X\in \{S,I,R,W\}$ in location $i$. If there are no other processes affecting compartment $i$, then
\begin{equation}\label{eq:inflow_outflow_metapop_X}
    \frac{d}{dt}X_i = \sum_{\stackrel{j=1}{j\neq i}}^{n}m^{X}_{ij} X_j
- \left(\sum_{\stackrel{j=1}{j\neq i}}^{n}m^{X}_{ji}\right) X_i
\end{equation}
describes the rate of change of $X_i$ because of the spatial coupling of $X_i$ for all other locations, where $m_{ij}^X$ is the rate of movement of compartment $X$ from location $j$ to location $i$. In order to simplify notation, we write
\[
m_{ii}^X = - \sum_{\stackrel{j=1}{j\neq i}}^{n}m^{X}_{ji}
\]
the rate of outflow from location $i$. 
Using this notation, \eqref{eq:inflow_outflow_metapop_X} is written as
\[
\frac{d}{dt}X_i = \sum_{j=1}^{n}m^{X}_{ij} X_j.
\]
The flow diagram in Figure~\ref{fig:flow_diagram} shows the model structure for two locations.
The mathematical model is then given in patch $i=1,\ldots,n$ by
\begin{subequations}
\label{sys:general_form}
\begin{align}
    \frac{d}{dt}S_i &= 
        b_i+\varepsilon_iR_i-\lambda_i S_i
        -\left(v_i+d^H_i\right)S_i+\sum_{j=1}^{n}m^{S}_{ij} S_j 
        \label{sys:general_form_dS} \\ 
    \frac{d}{dt}I_i &= 
        \lambda_i S_i-\left(\gamma_i+\delta_i+d^H_i\right)I_i 
        +\sum_{j=1}^{n}m^{I}_{ij} I_j 
        \label{sys:general_form_dI} \\ 
    \frac{d}{dt}R_i &=
        \gamma_i I_i+v_i S_i-\left(\varepsilon_i+d^H_i\right)R_i +\sum_{j=1}^{n}m^{R}_{ij} R_j 
        \label{sys:general_form_dR} \\ 
    \frac{d}{dt}W_i &= 
        \zeta_iI_i-d^C_iW_i
        +\sum_{j=1}^{n}m^{W}_{ij} W_j,
        \label{sys:general_form_dW} 
\end{align}
where 
\begin{equation}
    \lambda_i=\beta_i^II_i+\beta_i^WW_i
\end{equation}
is the force of infection acting on humans in patch $i$.
System \eqref{sys:general_form} is considered together with the initial condition
\begin{equation}
    S_i(0)>0, I_i(0) \geq 0, R_i(0)\geq 0, W_i(0)\geq 0,
\end{equation}
in which, to avoid a trivial situation, it is further assumed that there exists $i\in\{1,\ldots,n\}$ such that $I_i(0)+W_i(0)>0$.
\end{subequations}
State variables and parameters are summarised in Tables~\ref{tab:variable} and~\ref{tab:parameters}, respectively.
For convenience, denote $N_i^H=S_i+I_i+R_i$ the total number of (human) individuals in location $i=1,\ldots,n$.

\begin{table}[htbp]
\centering
\begin{tabular}{cl}
\toprule
Variable & Description \\
\midrule
$S_i$ & Number of susceptible individuals \\
$I_i$ & Number of infected individuals \\
$R_i$ & Number of recovered individuals \\
$N_i^H$ & Total number of individuals \\
$W_i$ & Concentration of \emph{V. cholerae} in water \\
\bottomrule
\end{tabular}
\caption{State variables of the SIRW model, where $i=1,\ldots,n$ is the index of the location.}
\label{tab:variable}
\end{table}

\begin{table}[htbp]
\centering
\begin{tabular}{cl}
\toprule
Parameter & Description\\ 
\midrule
\multicolumn{2}{l}{Human-related parameters} \\
$b_i$ & Recruitment rate of human population \\
$d_i^H$ & Natural death rate of humans \\
$v_i$ & Vaccination rate \\
$\delta_i$ & Disease-induced death rate \\
$\gamma_i$ & Recovery rate \\
$\varepsilon_i$ & Rate of loss of immunity \\
\midrule
\multicolumn{2}{l}{\emph{V. cholerae}-related parameters} \\
$\zeta_i$ & Pathogen shedding rate by humans \\
$d_i^C$ & Death rate of pathogens in water reservoir \\
\midrule
\multicolumn{2}{l}{Disease transmission-related parameters} \\
$\lambda_i$ & Force of infection acting on susceptible humans\\
$\beta_i^W$ & Infection by water \\
$\beta_i^I$ & Infection by humans \\
\midrule
\multicolumn{2}{l}{Movement-related parameters} \\
$m_{ij}^X$ & Movement rate of population $X$ from $j$ to $i$ \\ 
\bottomrule
\end{tabular}
\caption{Parameters of the SIRW model. The recruitment rate is constant, all other \emph{rates} are \emph{per capita}.}
\label{tab:parameters}
\end{table}

\subsection{Vector form of the model}
Before proceeding further, let us specify some vector and matrix notation used later. 
Vectors are assumed to be column vectors but are written without orientation unless required.
If $\bx=(x_1,\ldots,x_n)\in\IR^n$, then $\bx\geq\b0$ if $x_i\geq 0$, $\bx>\b0$ if $\bx\geq\b0$ and there exists $i$ such that $x_i>0$; finally, $\bx\gg\b0$ if $x_i>0$ for all $i=1,\ldots,n$. For $\bx,\by\in\IR^n$, $\bx\geq\by$, $\bx>\by$ and $\bx\gg\by$ if, respectively, $\bx-\by\geq 0$, $\bx-\by>0$ and $\bx-\by\gg\b0$.
The same notation is used for matrices.

Writing \eqref{sys:general_form} in vector form is extremely useful in the remainder of the work, so we introduce some notation here.
For $X\in\{S,I,R,W\}$, denote $\bX=(X_1,\ldots,X_n)$ and let
\[
\bZ = (\bS,\bI,\bR,\bW)
\]
be the vector of state variables.
Let $\bb=(b_1,\ldots,b_n)$,
$\bv=\diag(v_1,\ldots,v_n)$, 
$\bd^H=\diag\left(d_1^H,\ldots,d_n^H\right)$,
$\bgamma=\diag(\gamma_1,\ldots,\gamma_n)$, 
$\bdelta=\diag(\delta_1,\ldots,\delta_n)$, 
$\bvarepsilon=\diag(\varepsilon_1,\ldots,\varepsilon_n)$,
$\bzeta=\diag(\zeta_1,\ldots,\zeta_n)$,
$\bd^C=\diag\left(d_1^C,\ldots,d_n^C\right)$,
$\bbeta^I=\diag(\beta_1^I,\ldots,\beta_n^I)$ and 
$\bbeta^W=\diag(\beta_1^W,\ldots,\beta_n^W)$.
Finally, for $X\in\{S,I,R,W\}$, the \emph{movement matrix} for compartment $X$ is given by
\begin{equation}\label{eq:movement}
    \mathcal{M}^X=\left( \begin{array}{cccc}
        -\sum\limits^{n}_{k=1}m_{k1}^{X}	& m_{12}^{X} & \cdots & m_{1p}^{X}\\
        m_{21}^{X}	& -\sum\limits^{n}_{k=1}m_{k2}^{X} & \cdots  & m_{2p}^{X}\\
        \vdots	&  & \ddots & \vdots\\ 
        m_{n1}^{X} & m_{n2}^{X} & &  -\sum\limits^{n}_{k=1}m_{kn}^{X}
    \end{array}\right).
\end{equation}
Movement matrices of the form \eqref{eq:movement} enjoy a wide variety of properties useful to the analysis of metapopulation systems \cite{ArinoBajeuxKirkland2019}, which will be used later in the analysis.
Denote $\circ$ the Hadamard product.
Then the vector form of \eqref{sys:general_form} is
\begin{subequations}
	\label{sys:vector-form}
	\begin{align}
		\frac{d}{dt}\bS &= 
		\bb+\bvarepsilon\bR-\blambda\circ\bS
		-\left(\bv+\bd^H\right)\bS+\M^S\bS 
		\label{sys:vector-form-dS} \\ 
		\frac{d}{dt}\bI &= 
		\blambda\circ\bS-\left(\bgamma+\bdelta+\bd^H\right)\bI 
		+\M^I\bI 
		\label{sys:vector-form-dI} \\ 
		\frac{d}{dt}\bR &=
		\bgamma\bI+\bv\bS-\left(\bvarepsilon+\bd^H\right)\bR +\M^R\bR 
		\label{sys:vector-form-dR} \\ 
		\frac{d}{dt}\bW &= 
		\bzeta\bI-\bd^C\bW + \M^W\bW,
		\label{sys:vector-form-dW} 
	\end{align}
with force of infection
\begin{equation}
	\blambda = \bbeta_I\bI + \bbeta_W\bW.
\end{equation}
\end{subequations}
Note that we could also write $\blambda=\diag\left(\bbeta_I\bI + \bbeta_W\bW\right)$, thereby avoiding the use of the Hadamard product $\blambda\circ\bS$, but since the vector form \eqref{sys:vector-form} is used only when $\bI=\b0$ and therefore $\blambda=\b0$, this does not matter.

\section{Mathematical analysis of the model}
\label{sec:analysis}

\subsection{The model is well-posed}
First, since this is a substantial variation on \cite{tien2010multiple}, we need to ensure that the model makes sense as an epidemiological model: its solutions exist uniquely, the positive orthant is invariant under the flow of the system and solutions are bounded.
Existence and uniqueness follows evidently from the fact that the differential equations in \eqref{sys:general_form} are multivariate polynomials and hence $C^\infty$.
Recall that we denote $\bZ(t)$ the vector of state variables.
\begin{lemma}\label{lemma:invariance}
The nonnegative orthant $\mathbb{R}^{4n}_{+}$ is positively invariant under the flow of \eqref{sys:general_form}: if $\bZ(0)>\b0$ with $\bS(0)\gg\b0$, then $\bZ(t)\geq\b0$ and $\bS(t)\gg 0$ for all $t\geq 0$.
\end{lemma}
\begin{proof}
Assume initial conditions are $\bZ(0)>\b0$ with $\bS(0)\gg\b0$.
Let $t_1$ be the first time at which $X_k(t_1)<0$, for some $k\in K\subseteq\{1,\ldots,n\}$ and $X\in\{S,I,R,W\}$.
The question then is: which state variable exits the nonnegative orthant first?

Let us start by showing that, for the stated initial conditions, $S$ components cannot be the first to leave the nonnegative orthant; they cannot even limit on its corresponding faces.
Start with \eqref{sys:general_form_dS} and recall that it is assumed that $\bS(0)\gg\b0$.
Assume that there exist $K\subseteq\{1,\ldots,n\}$ and $t_1>0$ such that $S_{k}(t_1)= 0$ for $k\in K$, with $t_1$ the first time at which they vanish.
Then, for each $k\in K$, there are two cases: either $S_k$ becomes negative at $t_1$ (and then $S'_k(t_1)<0$) or $S_k$ becomes equal to zero at $t_1$, for which it first must become tangent to $0$ as $t\to t_1$ from the left, i.e., $S'_{k}(t)<0$ for $t\to t_1$ from the left. 
Excluding the latter also excludes the former, so let us consider tangency.
Since the $S_k$ are the first $S$ to become $0$, $S_{i}(t)>0$ for all $i\not\in K$ as $t$ approaches $t_1$ from the left.
Then, as $t$ approaches $t_1$ from the left, it follows from the nonnegativity of initial conditions and strong positivity of those of $\bS$, and the fact that $t_1$ is the first time at which any solution leaves the nonnegative orthant, that \eqref{sys:general_form_dS} takes the form
\[
S'_{k}(t)= b_k+\varepsilon_k R_k(t) +\sum_{\stackrel{j=1}{j\neq k}}^{n}m^{S}_{kj} S_j(t) > 0, 
\]
which contradicts the assumption $S'_{k}(t)<0$ as $t\to t_1$ from the left. 
Note that this is true even if $R_k$ were to hit their corresponding orthant faces at the same time $t_1$.
Therefore, $S_k$ cannot become tangent to $0$ at $t_1$ and the $S$ components cannot be the first to leave $\IR_+^{4n}$.

So now suppose there exists $K\subseteq\{1,\ldots,k\}$ and $t_1>0$ such that $I_{k}(t_1)= 0$, $I'_{k}(t_1)<0$, and $I_{i}(t_1)>0$ for all $i\not\in K$.
From the previous point, $S$ components cannot be the first to leave the nonnegative orthant, but they could do so at the same time as, now, the $I_k$ for $k\in K$; in any event, this means that $S_k(t_1)\geq 0$.
From \eqref{sys:general_form_dI}, it then follows that
\begin{equation*}
I'_{k}(t_1)= \beta_{W_k}W_k(t_1) S_k(t_1) +\sum\limits_{\stackrel{j=1}{j\neq k}}^{n}m^I_{kj}I_j(t_1) \geq 0,
\end{equation*}
which contradicts $I'_{k}(t_1)<0$. 
Thus, $I$ components cannot be the first to leave $\IR_+^{4n}$ either. By a similar reasoning, neither can the $R$ or $W$ components.
As a consequence, a solution starting in the nonnegative orthant remains there. Furthermore, $\bS(t)\gg 0$ for all $t\geq 0$.
\end{proof}

\begin{lemma}\label{lemma:boundedness-1}
The total human and bacterial populations are bounded.
\end{lemma}
\begin{proof}
Consider the evolution of the total human population $T_H=\sum_{i=1}^n N_i^H$ in \eqref{sys:general_form}.
It is easy to check that all sums describing movement cancel out.
Mechanistically speaking, the rate at which $S$ individuals move from, say, patch 1 to patch 2 has a negative sign when they leave patch 1 and a positive sign when they arrive in patch 2.
Since there is no death during transport and that transport is instantaneous, this implies that these quantities cancel out.
The same is true of $I$ and $R$ individuals.
Thus,
\[
    T_H'=\sum_{i=1}^n b_i -\sum_{i=1}^n d^H_iN_i -\sum_{i=1}^n \delta_iI_i \leq \sum_{i=1}^n b_i -\sum_{i=1}^n d^H_iN_i.
\]
Suppose $d^H_i>0$ for all $i$ (which is quite natural). 
Let $d^H_{min}=\min\limits_{i=1,\ldots,n}\bd^H_i$, then
\[
T_H'\leq \sum_{i=1}^n b_i -d^H_{min}N_H,
\]
which implies
\[
    T_H\leq e^{-td^H_{min}}\left[ N_H(0)+ B \int_{0}^{t} e^{d^H_{min} \tau } d\tau\right] 
    \leq e^{-t d^H_{min}}\left[ N_H(0)-\frac{1}{d^H_{min}}B \right]+\frac{B}{d^H_{min}},
\]
with $B =\sum_{i=1}^n b_i$.
So
\begin{itemize}
    \item[--] if $N_H(0)\leq B/d^H_{min}$ then $0<T_N(t)\leq B /d^H_{min}$,
    \item[--] if $N_H(0)>B/d^H_{min}$ then $T_H(t)\leq T_H(0)$ for all $t\in\IR_+$.
\end{itemize}
In other words, for all $t\in\IR_+$,
\begin{equation*}
 0\leq T_H(t)\leq \max \left(T_H(0),\dfrac{B}{d^H_{min}} \right).
\end{equation*}

Now denote $T_B=\sum_{i=1}^nW_i$ the total bacterial concentration. We have
\begin{align*}
    T'_B(t)=\sum_{i=1}^{n}\bzeta_i I_i-d_{min}^C T_B\leq K-d_{min}^C T_B, 
\end{align*}
because $I_i$ is bounded.
i.e.
\[
T_B(t)= e^{-t d^C_{min}}\left[ T_B(0)-\frac{K}{d^C_{min}} \right]+\frac{K}{d^C_{min}}.
\]
It follows that
\[
0\leq T_B\leq 
\max\left(T_B(0), \dfrac{K}{d_{min}^C}\right).
\]
By Lemma~\ref{lemma:invariance}, it follows that all variables are bounded.
\end{proof}
From the above, all solutions to \eqref{sys:general_form} eventually enter the set
\begin{equation}\label{eq:domain}
    \bOmega=\left\lbrace \bX=(\bS, \bI, \bR, \bW)\in \mathbb{R}^{4n}; 
    \quad 
    0\leq  T_H \leq \dfrac{B }{d^H_{min}}
    \text{ and } 
    0\leq T_B\leq 
\dfrac{K}{d_{min}^C}
 \right\rbrace,
\end{equation}
which is furthermore positively invariant under the flow of \eqref{sys:general_form}. 

One particular case of interest arises when the system has no disease or has the disease at a stable equilibrium and movement rates of humans are all equal.
\begin{proposition}
Assume that disease is at an equilibrium, i.e., $\bI=\bI^\star$.
Then, if movement rates are assumed to be equal for all disease states, i.e., $\mathcal{M} := \mathcal{M}^S = \mathcal{M}^I = \mathcal{M}^R$, then the total human population in each patch converges to $(\bd^H-\M)^{-1}(\bb-\bm{\delta}\bI^\star)$.
\end{proposition}

\begin{proof}
Let $\bN_H=\left(N_1^H,\ldots,N_n^H\right)$ be the vector of total human population in \eqref{sys:general_form}. 
Under the assumption of equal movement rates, we have
\begin{equation}
    \label{eq:dNHdt}
    \frac{d}{dt}\bN_H = \bb + \M\bN_H - \bd^H\bN_H-\bdelta\bI.
\end{equation}
By the assumption that there is no disease or the disease is at a stable equilibrium, we can substitute $\bI^\star$ for $\bI$, with $\bI^\star=\b0$ if there is no disease.
Then \eqref{eq:dNHdt} takes the form
\begin{equation}\label{eq:dNHdt_2}
\frac{d}{dt}\bN_H = \bb-\bm{\delta}\bI^\star -(\bd^H -\M)\bN_H.
\end{equation}
This equation admits the solution: 
\begin{equation}
\bN_H(t) = e^{-(\bd^H-\M)t}\left[\bN_H(0)-{(\bd^H-\M)}^{-1}(\bb-\bm{\delta}\bI^\star) \right]+(\bd^H-\M)^{-1}(\bb-\bm{\delta}\bI^\star).
\end{equation}
So, $\bN_H(t)$ converges to $(\bd^H-\M)^{-1}(\bb-\bm{\delta}\bI^\star)$.
\end{proof}

\subsection{Case of isolated locations}
\label{subsec:analysis-single-location}
Considering the system in the absence of movement between patches is useful for various purposes. 
In particular, it allows the parametrisation of the model by establishing the reproduction number for patches in isolation.
Although the reproduction number for the whole system is more complex than individual reproduction numbers, this is still useful.

In the absence of movement, the model takes for form, for a given patch $i=1,\ldots,n$
\begin{subequations}
\label{sys:general_form_1}
\begin{align}
    \frac{d}{dt}S_i &= 
        b_i+\varepsilon_i R_i-\lambda_i S_i
        -(v_i+d_i^H)S_i
        \label{sys:general_form_dS_1} \\ 
    \frac{d}{dt}I_i &= 
        \lambda_i S_i-(\gamma_i+\delta_i+d_i^H)I_i 
        \label{sys:general_form_dI_1} \\ 
    \frac{d}{dt}R_i &=
        \gamma_i I_i+v_i S_i-(\varepsilon_i+d_i^H)R_i
        \label{sys:general_form_dR_1} \\ 
    \frac{d}{dt}W_i &= 
        \zeta_i I_i-d_i^C W_i,
        \label{sys:general_form_dW_1} 
\end{align}
where 
\begin{equation}
    \lambda_i=\beta_{i}^II_i+\beta_{i}^WW_i.
\end{equation}
\end{subequations}

Suppose that $I_i=W_i=0$. Then we have the disease-free equilibrium
\begin{equation}
    \bE_{0}^{(i)}:=(S_i^\star,I_i^\star,R_i^\star,W_i^\star)
    = \left(
        \dfrac{\varepsilon_i+d_i^H}{\varepsilon_i+v_i+d_i^H}\;\dfrac{b_i}{d_i^H}, 
        0, 
        \dfrac{v_i}{\varepsilon_i+v_i+d_i^H}\;\dfrac{b_i}{d_i^H}, 
        0 
    \right).
\end{equation}

To compute the basic reproduction number, we use the method of \cite{VdDWatmough2002}.
Let $\mathcal{F}_i = \left(\lambda_i S_i,0 \right)^T$ contain new infections in the infected classes $\{I_i,W_i\}$ and $\mathcal{V}_i$ contain other flows within and outside the infected classes $\{I_i,W_i\}$, with a negative sign, i.e.,
\[
\mathcal{V}_i=
\begin{pmatrix}
(\gamma_i+\delta_i +d_i^H)I_i \\ 
-\zeta_i I_i + d_i^C W_i    
\end{pmatrix}.
\]
The matrix of new infections $F_i$ and the transfer matrix between compartments $V_i$ are then the Jacobian matrices obtained by differentiating $\mathcal{F}_i$ and $\V_i$ with respect to the infected variables $\{I_i,W_i\}$ and evaluating at the disease-free equilibrium $\bE_0^{(i)}$, i.e.,
\begin{equation}\label{eq:F-matrix-1-patch}
    F_i=
    \begin{pmatrix}
        \beta_{i}^I S_i^{\star}	&  \beta_{i}^W S_i^{\star} \\
        0 &  0
    \end{pmatrix}, \quad  V_i=
    \begin{pmatrix}
        \gamma_i +\delta_i + d_i^H	&  0 \\
        - \zeta_i &  d_i^C
    \end{pmatrix}.
\end{equation}
We have
\[
V^{-1}=
\frac{1}{d_i^C\left(\gamma_i+\delta_i+d_i^H\right)}
\begin{pmatrix}
        d_i^C   &   0 \\
        \zeta_i &   \gamma_i +\delta_i + d_i^H
\end{pmatrix}
\]
and the next generation matrix is then given by 
\[ 
F_iV^{-1}_i
=
\begin{pmatrix}
    \dfrac{\left(d^C_i\beta_i^I+\zeta_i\beta_i^W\right)S^\star_i}
        {d^C\left( \gamma_i +\delta_i + d^H_i\right)}	& 
    \dfrac{\beta_i^W S^\star_i}{d^C_i} \\ 
    0 &  0
\end{pmatrix}.
\]
The basic-reproduction number for patch $i=1,\ldots,n$ in the absence of movement is the spectral radius of this matrix and therefore takes the form
\begin{equation}\label{eq:R0-isolated}
    \mathcal{R}_{0}^{(i)}= 
    \dfrac{d_i^C\beta_{i}^I +\zeta_i\beta_{i}^W}
    {d_i^C\left(\gamma_i +\delta_i + d_i^H\right)}
    \dfrac{\varepsilon_i+d_i^H}{\varepsilon_i+v_i+d_i^H}\;\dfrac{b_i}{d_i^H}.
\end{equation}
We do not derive these results specifically for the isolated location case, but Theorems~\ref{th:DFE-LAS} and \ref{th:E0-GAS} derived later for the general model with movement both apply to an isolated location, so that in the absence of connection to other locations, the disease-free equilibrium $\bE_0^{(i)}$ is globally asymptotically stable when $\R_0^{(i)}<1$ and unstable when $\R_0^{(i)}>1$.

\subsection{The disease-free equilibrium}
\label{sec:DFE-full-system}
We now return to the full system \eqref{sys:general_form} with movement between patches.
At the disease-free equilibrium (DFE), $\bI=\bW=\b0$. 
As a consequence, at the DFE, the vector form \eqref{sys:vector-form} of \eqref{sys:general_form} reduces to the linear system
\begin{equation}\label{sys:at_DFE}
\frac{d}{dt}
\begin{pmatrix}
    \bS \\ \bR
\end{pmatrix}
=
\begin{pmatrix}
    \bb \\ \b0
\end{pmatrix} +
\begin{pmatrix}
    \M^S-\bv-\bd^H & \bvarepsilon \\
    \bv & \M^R-\bvarepsilon-\bd^H
\end{pmatrix}
\begin{pmatrix}
    \bS \\ \bR
\end{pmatrix}.
\end{equation}
Solving for equilibria of \eqref{sys:at_DFE} means finding $\bS^\star$ and $\bI^\star$ such that
\begin{equation}\label{sys:DFE}
\bm{\Sigma}_{DFE}
\begin{pmatrix}
    \bS^\star \\ \bR^\star
\end{pmatrix}
=
\begin{pmatrix}
    \bb \\ \b0
\end{pmatrix},
\end{equation}
where
\begin{equation}
\label{eq:matrix-for-DFE}
    \bm{\Sigma}_{DFE}:=\begin{pmatrix}
    \bv+\bd^H-\M^S & -\bvarepsilon \\
    -\bv & \bvarepsilon+\bd^H-\M^R
    \end{pmatrix}.
\end{equation}
This nonhomogeneous linear system has a unique solution if $\bm{\Sigma}_{DFE}$ is invertible.
This is guaranteed by \cite[Proposition 3]{ArinoBajeuxKirkland2019}, with $\bm{\Sigma}_{DFE}^{-1}>\b0$ if $\diag(\bm{v}+\bm{d}^H)\gg\b0$ and $\bm{\Sigma}_{DFE}^{-1}\gg\b0$ if $\diag(\bm{v}+\bm{d}^H)>\b0$ but $\M^S$ and $\M^R$ irreducible.

Using the formula for the inverse of a block $2\times 2$ matrix \cite{HornJohnson2013}, we then obtain an expression for the DFE.
We have
\[
\bm{\Sigma}_{DFE}^{-1}
= \begin{pmatrix}
    \bm{D}_{11} & \bm{\star} \\ \bm{D}_{21} & \bm{\star}
\end{pmatrix},
\]
with
\[
\bm{D}_{11} = 
\left(
\bv+\bd^H-\M^S-\bvarepsilon\left(\bvarepsilon+\bd^H-\M^R\right)^{-1}\bv
\right)^{-1}
\]
and
\[
\bm{D}_{21} = 
\left(\bvarepsilon+\bd^H-\M^R\right)^{-1}\bv\bm{D}_{11}.
\]
(The other two blocks are not shown as they do not play a role in computing the equilibria.)
From the previous discussion, all $\bm{D}_{ij}$ matrices are positive.
Hence,
\begin{equation*}
\bS^\star= \left(
\bv+\bd^H-\M^S-\bvarepsilon\left(\bvarepsilon+\bd^H-\M^R\right)^{-1}\bv
\right)^{-1}
\bb
\end{equation*}
and
\begin{equation*}
\bR^{\star}=\left(\bvarepsilon+\bd^H-\M^R\right)^{-1}\bv
\left(
\bv+\bd^H-\M^S-\bvarepsilon\left(\bvarepsilon+\bd^H-\M^R\right)^{-1}\bv
\right)^{-1}
\bb.
\end{equation*}
Thus,
\begin{equation}\label{eq:DFE}
\bm{E}^\star_{0}= \left(\bS^\star,\b0_{\mathbb{R}^n},\bR^\star,\b0_{\mathbb{R}^n}\right) \in \mathbb{R}^{4n}
\end{equation}
is the disease-free equilibrium of \eqref{sys:general_form}.

\subsection{Basic reproduction number}
Let $\mathcal{F} = (\lambda_1S_1,...,\lambda_nS_n ,0,..., 0)^T $ represent new infections in infected classes $I_1, ..., I_n, W_1, ..., W_n$ and $\mathcal{V}$ represent other flows within and outside the infected classes $I_1, ..., I_n, W_1,..., W_n$ (note that $\mathcal{V}$ has a negative sign).

We recall that the basic reproduction number is the average number of secondary infections produced by an infected individual when introduced into a population of susceptible individuals. In our case, $\mathcal{V}$ is given by:
\[\mathcal{V}=-\left( \begin{array}{c}
		-(\gamma_1+\delta_1 +d^H_1)I_1 +\sum\limits_{j=1}^{n}m^I_{1j}I_j\\
        \vdots\\
		-(\gamma_n+\delta_n+d^H_n)I_n +\sum\limits_{j=1}^{n}m^I_{nj}I_j\\ 
		\zeta_1 I_1 - d^C_1W_1+\sum\limits_{j=1}^{n}m^{W}_{1j} W_j\\
        \vdots\\
		\zeta_n I_n -d^C_1W_n+\sum\limits_{j=1}^{n}m^{W}_{nj} W_j\\
	\end{array}\right). \]
The matrix of new infections 
$F$ and the transfer matrix between compartments 
$V$ are then the Jacobian matrices obtained by differentiating $\mathcal{F}$ and 
$\V$ with respect to the infected variables and evaluating at the disease-free equilibrium, i.e.,
\[
		 F=D\mathcal{F}(\bE^{\star}_{0})=
         \left.\dfrac{\partial\lambda_i S_i}{\partial (I_i,W_i)}\right|_{\bE^\star_{0}} \text{ and }  V=D\mathcal{V}(\bE^{\star}_{0}).
\]
We have
	\begin{equation*} F=\begin{pmatrix}
			\dfrac{\partial\lambda_1S^\star_1}{\partial I_1}	& \hdots&\dfrac{\partial\lambda_1S^\star_1}{\partial I_n}&\dfrac{\partial\lambda_1S^\star_1}{\partial W_1}& \hdots&\dfrac{\partial\lambda_1S^\star_1}{\partial W_n}	\\	
			\vdots&   & \vdots& \vdots& & \vdots\\
			\dfrac{\partial\lambda_nS^\star_n}{\partial I_1}	& \hdots&\dfrac{\partial\lambda_nS^\star_n}{\partial I_n}&\dfrac{\partial\lambda_nS^\star_n}{\partial W_1}& \hdots&\dfrac{\partial\lambda_nS^\star_n}{\partial W_n}\\
			0& \hdots& 0 & 0 & \hdots&0\\
			\vdots &  & \vdots & \vdots& & \vdots\\
			0& \hdots& 0 & 0 & \hdots&0
		\end{pmatrix},
	\end{equation*}
where $\bS^{\star}=(S^{\star}_1,\ldots,S^{\star}_n)^T$ is the value of $\bS$ at the disease-free equilibrium.
Since $\partial\lambda_i/\partial I_j$ and $\partial\lambda_i/\partial W_j$ are only nonzero and equal to $\beta_{I_i}$ and $\beta_{W_i}$, respectively, when $i=j$, we have
\begin{equation}\label{eq:F_matrix_comptact}
    F=
    \begin{pmatrix}
        F_{11}	& F_{12} \\
        \mathbf{0} &  \mathbf{0} 
    \end{pmatrix}
:=			
    \begin{pmatrix}
    \diag(\beta_{I_1}S^{\star}_1,\ldots,\beta_{I_n}S^{\star}_n) &
    \diag(\beta_{W_1}S^{\star}_1,\ldots,\beta_{W_n}S^{\star}_n) \\
    \mathbf{0} & \mathbf{0}
    \end{pmatrix}
\end{equation}
and
\begin{equation} \label{eq:V_matrix}
V= \left( \begin{array}{cc}
			V_{11}	& \mathbf{0}	\\
			V_{21} &  V_{22}
\end{array}\right)
:=
\left( \begin{array}{cc}
		\diag(\gamma_i+\delta_i+d^H_i)	& \mathbf{0} \\  
			-\diag(\zeta_i) &  \diag(d^C_i)
\end{array}\right)- \left( \begin{array}{cc}
			\mathcal{M}^{I}	& \mathbf{0}	\\ 
			\mathbf{0} &  \mathcal{M}^{W}
		\end{array}\right).
\end{equation}
  
Since matrices $V_{11}$ and $ V_{22}$ are $n\times n$ irreducible M-matrices,  they have positive inverses.
The basic reproduction number is thus defined as the spectral radius
\[
\mathcal{R}_0 = \rho\left(FV^{-1}\right).
\]
Since $V$ is block lower triangular, we have
\[
V^{-1} =
\begin{pmatrix}
    V^{-1}_{11} & \mathbf{0} \\
    V^{-1}_{21} & V^{-1}_{22}
\end{pmatrix},
\]
where the blocks $V_{ij}^{-1}$ are computed using the formula for inverses of block $2\times 2$ matrices.
It follows that
\[
FV^{-1}= 
\begin{pmatrix}
    F_{11}	& F_{12} \\ 
    \mathbf{0} &  \mathbf{0} 
\end{pmatrix}
\begin{pmatrix}
    V^{-1}_{11} & \mathbf{0} \\
    V^{-1}_{21} & V^{-1}_{22}
\end{pmatrix}
= 
\begin{pmatrix}
    F_{11}{V}_{11}^{-1}+F_{12}{V}_{21}^{-1}	& F_{12}{V}_{22}^{-1}	\\
    \mathbf{0} &  \mathbf{0} 
\end{pmatrix}.
\]
From this, $\R_0=\rho(FV^{-1})=\rho\left(F_{11}{V}_{11}^{-1}+F_{12}{V}_{21}^{-1}\right)$, so we need only compute $V_{11}^{-1}$, $V_{21}^{-1}$ and $V_{22}^{-1}$ (since the latter plays a role in the first two).
Clearly, $V_{11}^{-1}=(\diag(\gamma_i+\delta_i+d_i^H)-\M^I)^{-1}$ and $V_{22}^{-1}=-(\diag(d^C_i)-\M ^W)^{-1}$, 
while $V_{21}^{-1}=(\diag(d^C_i)-\M^W)^{-1}\ \diag(\zeta_i)\ (\diag(\gamma_i+\delta_i+d_i^H)-\M^ I)^{-1}$.
It follows that 
\begin{multline}\label{eq:R0-global}
\mathcal{R}_0 = \rho \bigg(
\diag(\beta_{I_i}S_i^\star)
\left(\diag(\gamma_i+\delta_i+d_i^H)-\M^I\right)^{-1}+\\
\diag(\beta_{W_i}S_i^\star)
\left(\diag(\zeta_i)-\M^W\right)^{-1}
\diag(\zeta_i)
\left(\diag(\gamma_i+\delta_i+d_i^H)-\M^ I\right)^{-1}
\biggr).
\end{multline}

This provides an expression that can be evaluated when parameter values are chosen for each patch. 
It is possible to derive an explicit expression for $\R_0$ in some cases when the number of patches is low, but there is very little value in doing so and expression \eqref{eq:R0-global} is the one we use in practice.
\subsection{Stability of the disease-free equilibrium}
First of all, from \cite[Theorem 2]{VdDWatmough2002}, we have the following result.
\begin{theorem}
\label{th:DFE-LAS}
    $\bE^\star_0$ is locally asymptotically stable if $\mathcal{R}_0 < 1$ and unstable if $\mathcal{R}_0 >1$.
\end{theorem}

\begin{proof}
We need to check that the hypotheses A1-A5 of \cite[Theorem 2]{VdDWatmough2002} are satisfied.
Hypotheses A1-A4 follow from the procedure used above to derive $F$ and $V$ in the computation of $\R_0$.
Therefore, all we need to check is that the system without disease has the DFE locally asymptotically stable.
In the absence of disease, the vector form \eqref{sys:vector-form} of \eqref{sys:general_form} is the linear system,
\begin{align*} 
\dfrac{d}{dt}\bS &= \bb+(\M^S-\bv-\bd^H)\bS+\bvarepsilon\bR \\
\dfrac{d}{dt}\bR &= \bv\bS+(\M^R-\bvarepsilon-\bd^H)\bR.
\end{align*}
This is exactly \eqref{sys:at_DFE}, so we know it has a unique equilibrium, the disease-free equilibrium \eqref{eq:DFE}.
The Jacobian matrix of the system at any point takes the form
\[
\begin{pmatrix}
    \M^S-\bv-\bd^H & \bvarepsilon \\
    \bv & \M^R-\bvarepsilon-\bd^H
\end{pmatrix}.
\]
We saw in Section~\ref{sec:DFE-full-system} that this matrix is invertible.
From \cite[Lemma 2 and Proposition 3]{ArinoBajeuxKirkland2019}, the spectral abscissa of this matrix is negative, so the disease-free equilibrium is always (locally) asymptotically stable, verifying that assumption (A5) in \cite[Theorem 2]{VdDWatmough2002} holds.
The result follows.
\end{proof}

The following Lemma reinforces Lemma~\ref{lemma:boundedness-1} and is useful when considering the behaviour when $\R_0<1$.

\begin{lemma}
\label{lemma:boundedness-2}
In \eqref{sys:general_form}, the population of susceptible individuals in isolated patch $i=1,\ldots,n$ is asymptotically bounded above, with upper bound given by
\begin{equation}\label{eq:upper_bound_Si}
    \dfrac{b_i}{v_i+d_i^H}\left(1+\frac{\varepsilon_i}{b_i}\right).
\end{equation}
\end{lemma}
\begin{proof}
First, note that from \eqref{eq:domain}, there holds that, for all $i=1,\ldots,n$,
\[
R_i\leq N_i^H\leq \dfrac{b_i}{d_i^H}.
\]
Using this, we have
\begin{align*}
    S'_i &= b_i+\varepsilon_i R_i-\lambda_i S_i-(v_i+d_i^H)S_i \\
    &\leq b_i+\varepsilon_i R_i-(v_i+d_i^H)S_i\\
    &\leq  b_i\left(1+\frac{\varepsilon_i}{d_i^H}\right)-(v_i+d_i^H)S_i.
\end{align*}
Therefore, passing to the limit,
\[
    S_i(t)\leq \dfrac{b_i}{v_i+d_i^H}\left(1+\frac{\varepsilon_i}{d_i^H}\right)
\]
for all sufficiently large $t$.
\end{proof}
\begin{theorem}\label{th:E0-GAS}
$\bE^{\star}_{\b0}$ is globally asymptotically stable when $\mathcal{R}_0^{(i)} < 1$ for all $i=1,\ldots,n$.
\end{theorem}

\begin{proof}  
Let us show that the function
\[
L(\bS,\bI,\bR,\bW)=\sum_{i=1}^{n} (I_i +  f_i W_i),
\]
with $f_i>0$, is a Lyapunov function for \eqref{sys:general_form}.
First, note that $L(\bS^\star,\b0,\bR^\star,\b0)=0$ and $L(\bS, \bI, \bR, \bW)\geq 0$. Hence $L$ is positive definite.
Then,
\begin{align*}
\frac{d}{dt}L(\bS, \bI, \bR, \bW) &= \sum_{i=1}^{n}I'_i+ \sum_{i=1}^{n}f_iW'_i \\
&= 
\sum_{i=1}^{n}\left[(\beta_{Ii} I_i+\beta_{Wi}W_i)S_i-(\gamma_i+\delta_i+d_i^H)I_i\right] \\
& \quad+\sum_{i=1}^{n}\left[\sum_{j=1}^{n}m_{ij}^{I}I_j+f_i\zeta_iI_i-f_i d_i^C W_i+f_i\sum_{j=1}^{n}m_{ij}^W W_j\right]\\
&= \sum_{i=1}^{n}\left[(\beta_{Ii} I_i+\beta_{Wi}W_i)S_i-(\gamma_i+\delta_i+d_i^H-f_i\zeta_i)I_i-f_id_i^CW_i\right]\\
&\leq \sum_{i=1}^{n}\left[(\beta_{Ii} I_i+\beta_{Wi}W_i)\dfrac{b^{\star}_i}{v_i+d_i^H}-(\gamma_i+\delta_i+d_i^H-f_i\zeta_i)I_i-f_id_i^CW_i\right] \\
&= \sum_{i=1}^{n}\left[\dfrac{\beta_{Ii}b^{\star}_i}{v_i+d_i^H}-(\gamma_i+\delta_i+d_i^H)+\zeta_i f_i \right]I_i+\sum_{i=1}^{n}\left[\dfrac{\beta_{Wi}b^{\star}_i}{v_i+d_i^H} -d_i^Cf_i \right] W_i,
\end{align*}
where $b_i^\star=b_i \left(1+\varepsilon_i/d_i^H\right)$ is used to simplify \eqref{eq:upper_bound_Si}.
We choose the constant $f_i$ such that
\[
    \dfrac{\beta_{Ii}b^{\star}_i}{v_i+d_i^H}-\left(\gamma_i+\delta_i+d_i^H\right)+\zeta_i f_i=0.
\]
i.e., 
\[
   f_i= \dfrac{ \left(v_i+d_i^H\right)\left(\gamma_i+\delta_i+d_i^H\right)-\beta_{Ii}b^{\star}_i}{\zeta_i\left(v_i+d_i^H\right)}.
\]
So 
\begin{align*}
    \frac{d}{dt}L(\bS, \bI, \bR, \bW) & \leq 
    \sum_{i=1}^{n}\left[
    \dfrac{b^{\star}_i\left(d_i^C\beta_{Ii}+\zeta_i\beta_{Wi}\right)}
    {d_i^C \left(v_i+d_i^H\right)\left(\gamma_i+\delta_i+d_i^H\right)}-1 
    \right]
    \dfrac{d_i^C\left(\gamma_i+\delta_i+d_i^H\right)}{\zeta_i}W_i\\
   & \leq 
   \sum_{i=1}^{n}\left[
   \left(\frac{b_i+\epsilon_i}{b_i}\right)\mathcal{R}_0^{(i)}-1 
   \right]
   \dfrac{d_i^C\left(\gamma_i+\delta_i+d_i^H\right)}{\zeta_i}W_i\\
   &\leq \sum_{i=1}^{n}\left[
   \mathcal{R}_0^{(i)}-\left(\frac{b_i}{b_i+\epsilon_i}\right) \right]
   \dfrac{d_i^C\left(b_i+\epsilon_i\right)\left(\gamma_i+\delta_i+d_i^H\right)}
   {\zeta_ib_i}W_i.
\end{align*}
Then $L'(\bS, \bI, \bR, \bW)$ is negative if $\mathcal{R}_0^{(i)}< 1$ for all $i=1,\ldots,n$.
This shows that $L$ is indeed a strict Lyapunov function for \eqref{sec:model} near $\bE_0$; the global asymptotic stability of the DFE follows by the LaSalle invariance principle.
\end{proof}
Thus, if all patches in isolation are such that $\R_0^{(i)}<1$, then the system-wide disease-free equilibrium $\bE_{0}^\star$ is globally asymptotically stable.
Thus, population and pathogen mobility does not make the situation any worse.
Note that this does not completely clarify the situation, since the result is expressed in terms of the individual $\R_0^{(i)}$ in locations given by \eqref{eq:R0-isolated} instead of the system-wide $\R_0$ given by \eqref{eq:R0-global}.
However, although this is not a general rule, metapopulation models often satisfy $\min_{i=1,\ldots,n}\R_0^{(i)}\leq \R_0\leq\max_{i=1,\ldots,n}\R_0^{(i)}$.
If this were the case here, then the conditions of Theorem~\ref{th:E0-GAS} would indeed imply global asymptotic stability when $\R_0<1$.
We were not able to derive a relationship between \eqref{eq:R0-isolated} and \eqref{eq:R0-global}, but considered the relationship numerically in Section~\ref{sec:computational-analysis-full-system}.

The situation when $\R_0>1$ is much more complicated to study.
We suspect that as is common with this type of model, when $\R_0>1$, there is a unique endemic equilibrium that is globally asymptotically stable.
However, this is only considered numerically in the present paper.
\section{Computational analysis}
\label{sec:computational}

\begin{table}[htbp]
    \centering
    \begin{tabular}{lllll}
        \toprule
        Parameter & Plausible range & Default value & Unit & Source \\
        \midrule
        \multicolumn{5}{c}{Human-related parameters} \\
        $b_{i}$ & -- & -- & people$\times$day$^{-1}$ & Computed \\ 
        $1/d_i^H$  & -- & 52.5 years & day & World Bank \cite{worldbank_life_expectancy} \\
        $1/\delta$ & [1,50] & 7 & day & \cite{canada_cholera}\\ 
        $1/\gamma$ & [2,60] & 10 & day & \cite{pasteur_cholera}\\
        $1/v$ & [10,730] & 100 & day &\cite{who_cholera_vaccination}\\ 
        $1/\varepsilon$ & [50 days,5 years] & 2 years & day & \cite{pasteur_cholera}\\
        \midrule
        \multicolumn{5}{c}{\emph{V. cholerae}-related parameters} \\
        $\zeta$ & [1,10000] & 1000 & cell L$^{-1}$ & \cite{tuite2011cholera} \\ 
$1/d_C$ & [0.5,30] & 2 & day & See text \\
        \midrule
	\multicolumn{5}{c}{Disease transmission-related parameters} \\
        $\mathcal{R}_0^{(i)}$ & $[0.8,10]$ & 3 & -- & \cite{tien2010multiple}\\
	$\beta_i^{I}$ & [1e-08,1e-06] & 1.101193e-07& day$^{-1}$ &\cite{tien2010multiple}\\ 
	$\beta_i^{W}$ & [1e-08,1e-06] & 1.215594e-07 & day$^{-1}$ & \cite{tien2010multiple} \\ 
        \midrule
        \multicolumn{5}{c}{Movement-related parameters} \\
        $m^H$ & [5,1000] people &  & day$^{-1}$ & See text \\
        $m^W$ & -- &  & day$^{-1}$ & See text \\
        \bottomrule	
    \end{tabular}
    \caption{Values of the parameters used in numerical work. Refer to Table~\ref{tab:parameters} for interpretation of the parameters.} 
    \label{tab:parameter-values}
\end{table}

Parameter ranges used in the computational analysis are given in Table~\ref{tab:parameter-values}.
Let us comment on the choice of some of the values.

We compute $b_i$ so that in the absence of disease, $b_i/d_i^H$ gives the population for the location under consideration. 
This is taken to be 10,000 individuals in most of the numerical investigations below; regarding average duration of life, we take the average 52.5 years duration of life at birth in Chad, which is the main focus of this work.

Regarding disease parameters, we allow all parameters to vary in large intervals, as the model can be used in a variety of circumstances, e.g., an outbreak in a refugee camp, which is characterised by high disease-induced death rate $\delta$ and high vaccination rate $v$, but also, ``routine'' vacination with a much lower rate.

We make use of the value of the basic reproduction number $\R_0$ to determine the values of the transmission parameters $\beta^W$ and $\beta^I$. See the discussion in Section~\ref{sec:computational-analysis-single-location}.
In terms of the allowable interval of values of the within-location reproduction numbers $\R_0^{(i)}$ or the system-wide value $\R_0$, we note that values vary depending on several factors, notably the geographical context, environmental conditions, access to potable water and population density.
In the literature, the mean estimate seems to be of $\R_0\in[1.1,4]$, but much higher values are observed in vulnerable populations such as those arising in refugee camps during humanitarian crises, when access to potable water is extremely limited \cite{codecco2001endemic, ganesan2020cholera, nkwayep2024prediction}.
This is the reason why we allow values of up to 20. 
We also consider the non-endemic case and thus consider reproduction numbers in the range $[0.5,20]$.

Finally, concerning movement rates, we use the approximation in \cite{ArinoPortet2015}, which we summarise here.
Suppose $x(t),y(t)$ are the counts or densities of individuals in two locations.
Suppose that we work over a period of time sufficiently short that other flows affecting the system can be neglected.
Denote $m_{yx}$ the movement rate from $x$ to $y$.
Then assuming movement is linear as in \eqref{sys:general_form}, we have $x'=-m_{yx}x$ and $y'=m_{yx}x$. 
Thus, we have the solutions $x(t)=x(0)\exp(-m_{yx}t)$ and $y(t)=y(0)\exp(m_{yx}t)$.
Let $\Delta_{yx}=x(0)-x(1)$, i.e, $\Delta_{yx}$ is the change in one time unit of the number or the concentration of individuals in location $x$.
It follows that
\[
\Delta_{yx}=x(0)-x(1)=x(0)\left(1-e^{-m_{yx}}\right),
\]
i.e., 
\[
m_{yx}=-\ln\left(1-\frac{\Delta_{yx}}{x(0)}\right).
\]
This can also be expressed in terms of the population or density in the destination by noting that in this case, $\Delta_{yx}=y(1)-y(0)=y(0)(e^{m_{yx}}-1)$.
As we consider movement in the range of $[5,1000]$ people a day from one location to another and that all locations have a population of 10,000, this means movement rates are in the range [0.0005,0.105].



\subsection{The model in an isolated location}
\label{sec:computational-analysis-single-location}
Metapopulations result from the coupling of similar units, so investigating these constituting units is an extremely useful first step. 
We start the computational analysis by considering a single location in the absence of connection with other locations. 
See the relevant Section~\ref{subsec:analysis-single-location} in the mathematical analysis.
Note that for simplicity, we denote $\beta_W$ and $\beta_I$ the contact parameters in this section.

  \begin{figure}[htbp]
    \centering
    \includegraphics[width=0.8\linewidth]{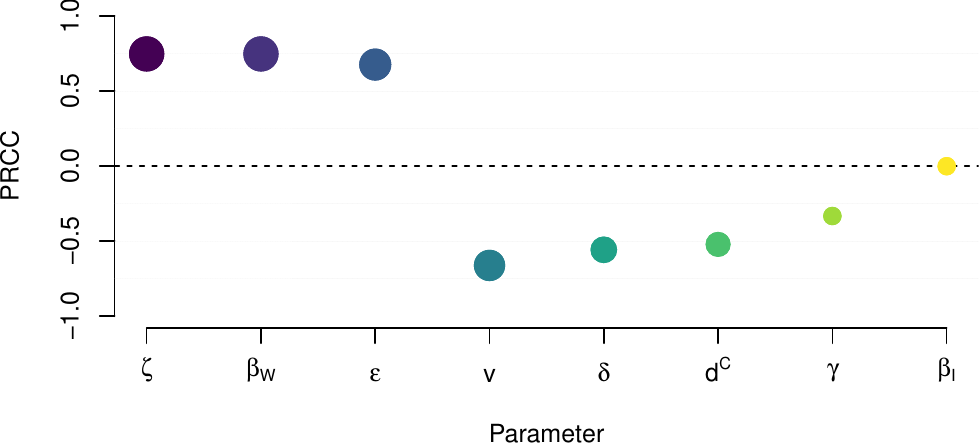}
    \caption{Partial rank correlation coefficients sensitivity analysis of $\R_0^{(i)}$ to coefficients in the isolated case.}
    \label{fig:sensitivity-R0-isolated}
\end{figure}

We first consider the sensitivity of $\R_0^{(i)}$ given by \eqref{eq:R0-isolated} to changes in model parameters.
As noted earlier, because there are two different values of $\beta$ involved, the approach that consists in setting a target range for $\R_0$ then deducing a range for $\beta$ does not work.
We use an \emph{ad-hoc} approach: we ``oversample'' (using Sobol sampling), using a very large number of parameters to cover the ranges of variation of the parameters and retain only those points for which $\R_0^{(i)}$ falls in the prescribed range in Table~\ref{tab:parameter-values}.
Figure~\ref{fig:sensitivity-R0-isolated}, for instance, was generated using a sample of 10 million points, of which 73.2\% gave $\R_0^{(i)}$ values in the desired interval of $[0.5,20]$.

In Figure~\ref{fig:sensitivity-R0-isolated}, we observe that the rate of loss of immunity $\varepsilon$ plays a role comparable in magnitude to that of shedding rate $\zeta$ and the infection parameter $\beta_W$.
Another interesting observation is that the parameter $\beta_I$ describing human-to-human transmission plays virtually no role in changes of $\R_0^{(i)}$, despite being allowed to vary in a comparable range to that of $\beta_W$.

\begin{figure}[htbp] 
    \centering 
    \begin{subfigure}{0.49\textwidth}
        \includegraphics[width=\textwidth]{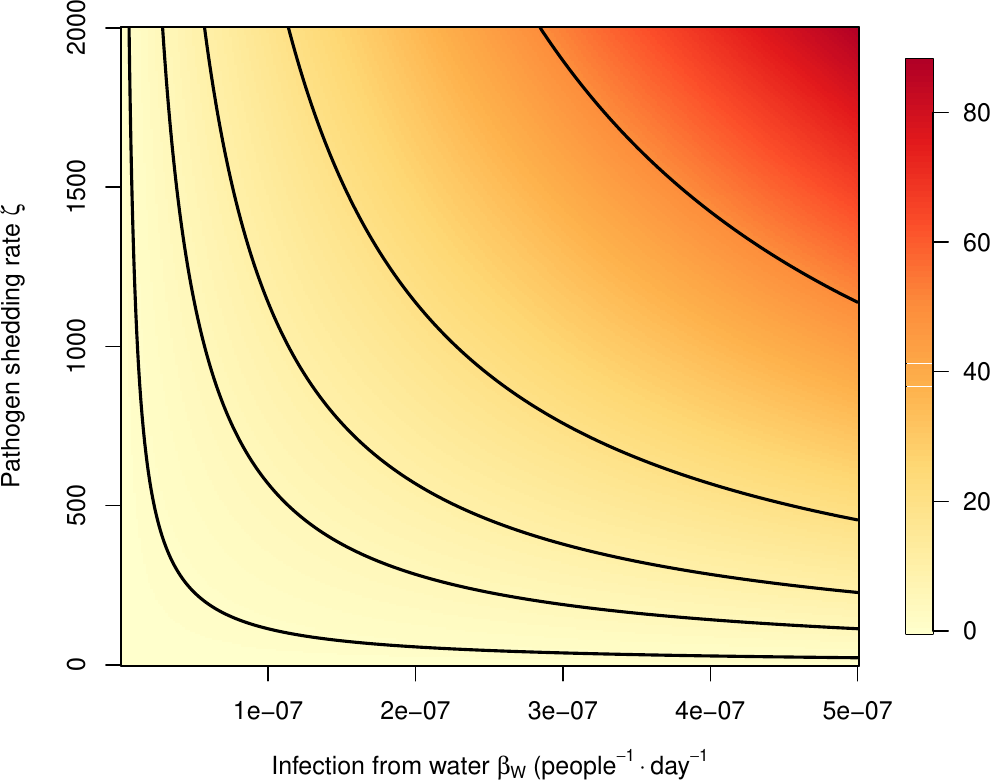} 
        \caption{$\beta_{W}$ and $\zeta$} 
        \label{fig:heatmap-metapo-R0-betaW-zeta} 
    \end{subfigure}
    \begin{subfigure}{0.49\textwidth}
        \includegraphics[width=\textwidth]{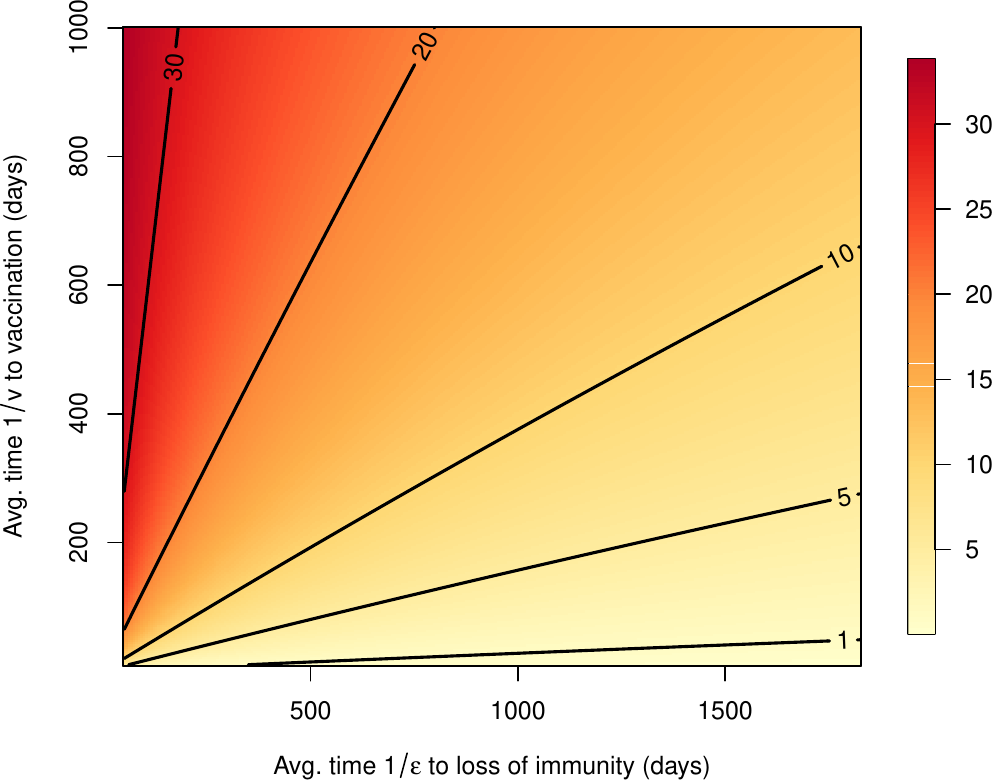} 
        \caption{$1/\varepsilon$ and $1/v$} 
        \label{fig:heatmap-metapo-R0-epsilon-v} 
    \end{subfigure}
    \caption{Value of $\R_0^{(i)}$ as a function of the variation of pairs of parameters.}
    \label{fig:heatmaps}
\end{figure} 

In order to better understand the relative roles of significant parameters, we examine a few combinations in heatmap plots (Figure~\ref{fig:heatmaps}), which show how $\R_0^{(i)}$ changes as a function of the values of pairs of parameters.
Comparing the roles of the contact parameter $\beta_W$ and the shedding rate $\zeta$ in Figure~\ref{fig:heatmap-metapo-R0-betaW-zeta}, we observe that their role is roughly equivalent, as could be inferred from Figure~\ref{fig:sensitivity-R0-isolated}.
From the perspective of disease control, this indicates that measures targetting the contamination of water through $\zeta$ and those targetting water uptake by humans through $\beta_W$ both have an effect of the same nature, so that there is no specific priority in terms of what mechanism to target.
This is quite different from what is seen in Figure~\ref{fig:heatmap-metapo-R0-epsilon-v}, which shows the effect of two parameters that can be thought of as representing aspects of vaccination, the mean time $1/\varepsilon$ to loss of immunity and the mean time $1/v$ to vaccination.
Indeed, recall that our simple model does not distinguish between disease-induced and vaccine-induced immunity, with $\varepsilon$ being simply the rate at which individuals leave the $R$ compartment and become susceptible to the disease again.
While the part of this flow concerning loss of naturally acquired immunity is a characteristic of cholera and cannot be acted upon, the part that concerns loss of vaccine-induced immunity can; making a longer lasting vaccine means decreasing $\varepsilon$.
In Figure~\ref{fig:heatmap-metapo-R0-epsilon-v}, we see that the mean time $1/\varepsilon$ to loss of immunity is the parameter having most effect when times to vaccination are near the middle of the range. 
However, increasing $1/varepsilon$ can only be done through long-term research and might actually not be possible without a radical change in vaccination technology.
In ``real life'', it can therefore be assumed that whatever its value is, $\varepsilon$ is a parameter that public health authorities have little control over. 
However, if the mean time to loss of immunity is indeed within the 2--3 years range as we hypothesise in Table~\ref{tab:parameter-values}, then it is clear that decreasing the mean wait-time to vaccination is potentially quite impactful on the reproduction number.

\subsection{The full interconnected system}
\label{sec:computational-analysis-full-system}
Returning to the full system, we consider three scenarios for the flow of water, which are illustrated in Figure~\ref{fig:graphs}.
In all cases, the locations are assumed to be close and thus the movement of humans is taken to be a complete graph.
\begin{figure}[htbp]
	\centering
	\def\hhskip{*2}
	\def\vvskip{*2}
	\begin{tikzpicture}[auto, auto,
		cloud/.style={draw, ellipse,fill=gray!20},
            scale=0.6, transform shape]
		\node [cloud] at (0\hhskip,-0.5\vvskip) (11) {$1$};
		\node [cloud] at (1\hhskip,-0.5\vvskip) (12) {$2$};
		\node [cloud] at (2\hhskip,-0.5\vvskip) (13) {$3$};
		\path [line, thick] (11) to (12);
		\path [line, thick] (12) to (13);
		\node [cloud] at (4\hhskip,0\vvskip) (21) {$1$};
		\node [cloud] at (3\hhskip,-1\vvskip) (22) {$2$};
		\node [cloud] at (5\hhskip,-1\vvskip) (23) {$3$};
		\path [line,thick] (21) to (22);
		\path [line,thick] (21) to (23);
		\path [line,thick,bend left] (22) to (23);
		\path [line,thick,bend left] (23) to (22);
		\node [cloud] at (6\hhskip,-0.5\vvskip) (31) {$1$};
		\node [cloud] at (7\hhskip,-0.5\vvskip) (32) {$2$};
		\node [cloud] at (8\hhskip,-0.5\vvskip) (33) {$3$};
		\path [line,thick,bend left] (31) to (32);
		\path [line,thick,bend left] (32) to (31);
		\path [line,thick,bend left] (32) to (33);
		\path [line,thick,bend left] (33) to (32);
		\node [cloud,fill=white,draw=white] at (1\hhskip,-1.75\vvskip) (a) {(a)};
		\node [cloud,fill=white,draw=white] at (4\hhskip,-1.75\vvskip) (b) {(b)};
		\node [cloud,fill=white,draw=white] at (7\hhskip,-1.75\vvskip) (c) {(c)};
	\end{tikzpicture}
	\caption{Three water flow scenarios: (a) river, (b) one patch upstream, two patches in the estuary, (c) three patches along the coastline of a lake.}
	\label{fig:graphs}
\end{figure}
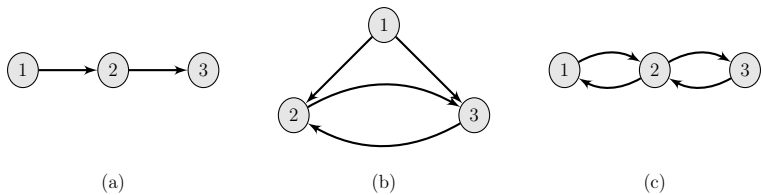
Scenarios (b) and (c) are tailored to the endorheic situation of the Lake Chad Basin, i.e., bacteria are not washed out by a flow of water.

\begin{figure}
    \centering
    \includegraphics[width=0.8\linewidth]{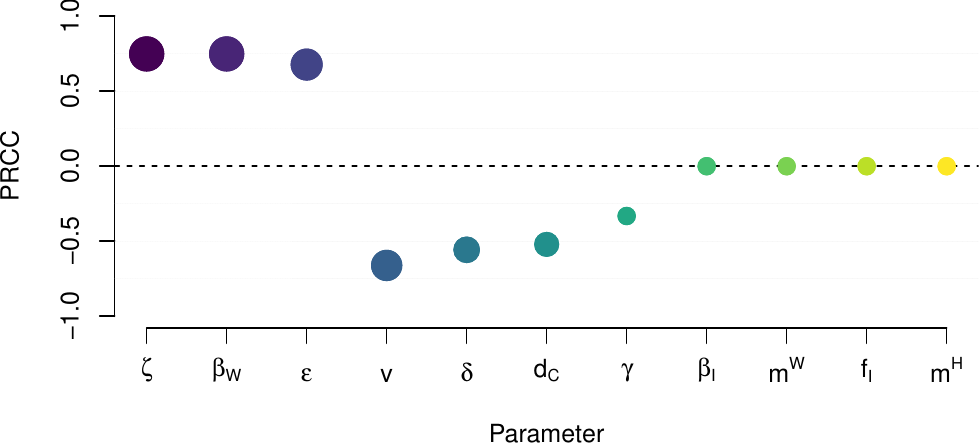}
    \caption{Sensitivity analysis of the system-wide $\R_0$ \eqref{eq:R0-global}.}
    \label{fig:sensitivity-R0-3patches-river}
\end{figure}
First we perform a sensitivity analysis of the system-wide $\R_0$ given by \eqref{eq:R0-global} to variations in parameters.
We proceed as in the isolated case, except that here, we add movement rates $m^H$ of humans and $m^W$ of water to the parameters being varied.
To simplify the sensitivity analysis, we use a single value for each type of movement, using the same value in all relevant off-diagonal entries in the movement matrices.
We add one movement parameter, $f_I\in[0,1]$, not in the mathematical model, which multiplies $\M^I$.
This allows to evaluate the effect of having a different rate of movement for infectious individuals compared to movement of other humans.
Like for movement rates, we use the same parameter value for all locations, e.g., $\beta_I$ is the same in all three locations.
We show the result of the sensitivity analysis for the river scenario in Figure~\ref{fig:sensitivity-R0-3patches-river}.

We observe that movement rates play very little role in the variation of the value of the system-wide $\R_0$.
In the sensitivity figures, values are ordered from left to right by decreasing absolute value of partial rank correlation coefficients, so $m^W$, $f_I$ and $m^H$ play a smaller role than $\beta_I$; all are very small, ranging from a PRCC of $4.079725\times10^{-5}$ for $m^W$, to a PRCC of $-1.632381\times10^{-5}$ for $m^H$, with the PRCC of $f_I$ being $-3.704921\times10^{-5}$.

Note that even though only Figure~\ref{fig:graphs}(c) is strongly connected, most cases will see disease present everywhere because the corresponding digraphs for human movement are complete.
\begin{figure}[htbp]
    \centering
    \includegraphics[width=0.65\linewidth]{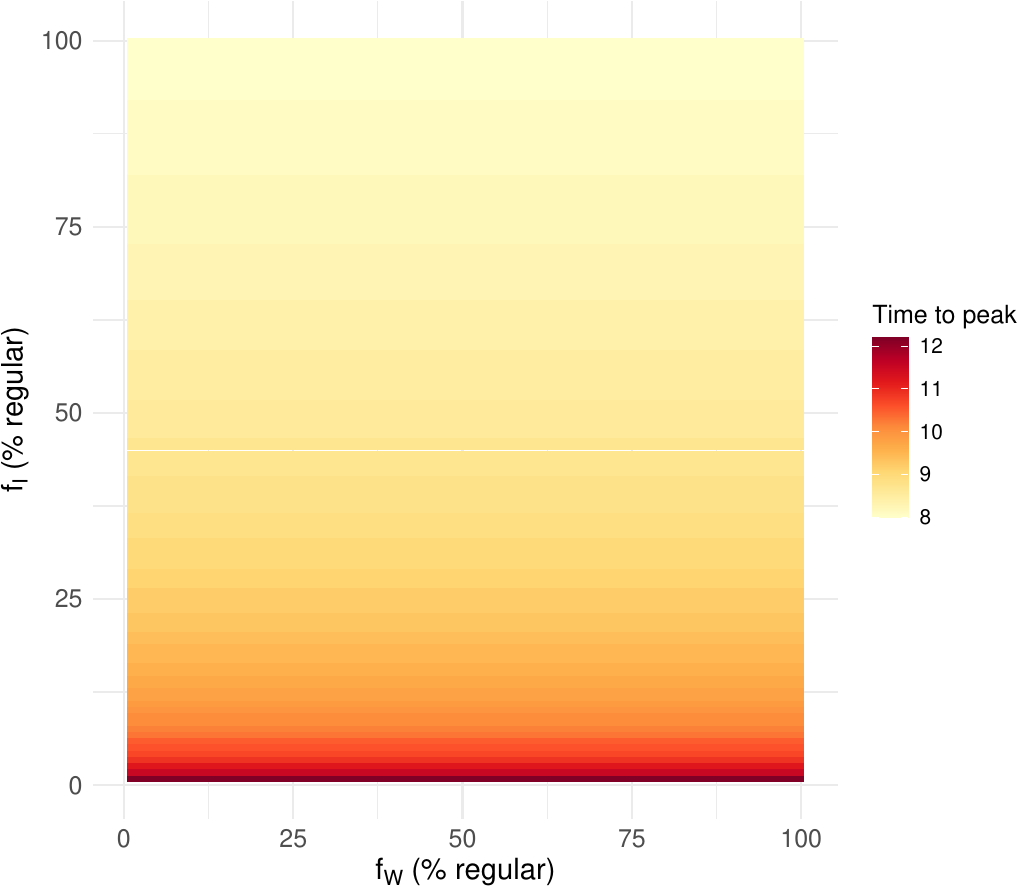}
    \caption{Time to peak prevalence in patch 3 in the river scenario of Figure~\ref{fig:graphs} as a function of the percentage $f_W$ and $f_H$ of nominal movement of water and humans, respectively.}
    \label{fig:heatmap-time-to-peak}
\end{figure}

Actually, the movement of humans ``regularises'' the spread by making it much faster when it is present.
We illustrate this in Figure~\ref{fig:heatmap-time-to-peak}, where we observe that as soon as humans are allowed to move, the time until a peak of prevalence initiated in patch 1 reaches patch 3 greatly diminishes.
Also, we note that the rate of movement of bacteria in the water does not play a role in this scenario.
We do not present results for the other scenarios: they are very similar to the one shown here.

\begin{figure}[htbp]
    \centering
    \includegraphics[width=0.7\linewidth]{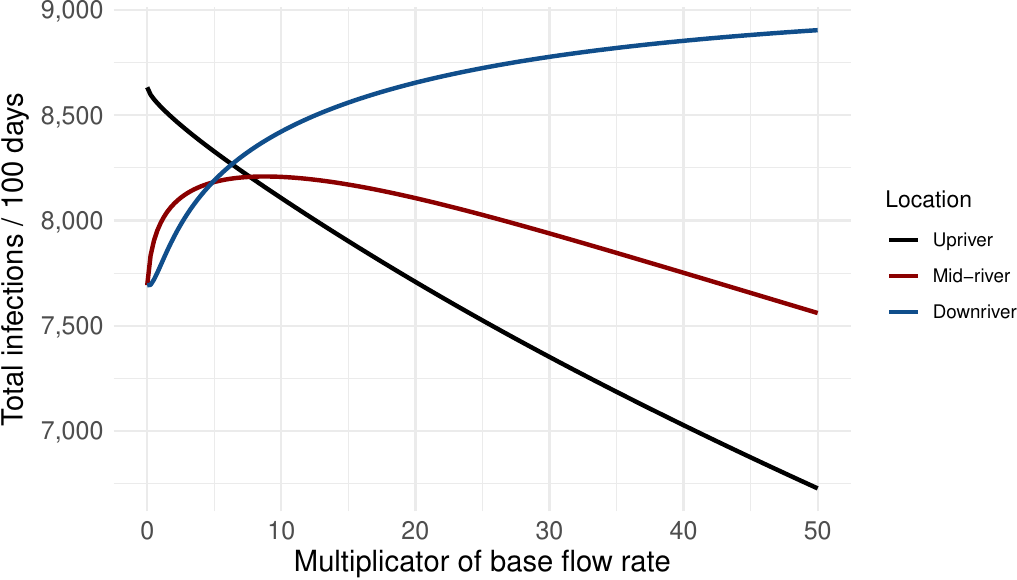}
    \caption{Total number of water-induced infections over 100 days in the river scenario, where locations are named upriver, mid-river and downriver. 
    Initial conditions are no infectious human in any location and an initial pathogen concentration in water of 500 units upriver and none in mid-river or downriver.}
    \label{fig:final-size-high-flow-river}
\end{figure}
We conclude this part by running a numerical experiment motivated by the following thought experiment. 
Given that human-to-human transmission plays very little role, in a river scenario, as the flow rate increases, locations upriver should see a drop in cases, because of a local ``washout'' situation as far as bacteria are concerned, while cases should increase downriver.

To investigate this, we take a (somewhat arbitrary) base water flow rate $m_{21}^W=m_{32}^W=0.04$ from upriver to mid-river and mid-river to downriver as well as an equal number of 50 humans moving between any pair of locations each day.
All other parameters are taken equal in all patches and at the default values given in Table~\ref{tab:parameter-values}.
We choose a combination of $\beta_{Ii}$ and $\beta_{Wi}$ so that $\R_0^{(i)}=5$ in all locations.
In Figure~\ref{fig:final-size-high-flow-river}, we show the total number $\int_0^{100}\beta_{Wi}S_i(s)W_i(s)ds$ of infections accrued over a simulation time period of 100 days in each location as the flow rate is varied from 0 to 50 times the base flow rate.

We observe that for the upriver location, the effect is immediate: the total number of infections decreases as soon as the outflow rate increases; the effect on the downriver location is likewise immediate, albeit in the other direction.
The mid-river location, on the other hand, has the interesting transition from an initial increase when the flow rate increases, to a peak then a decrease once the flow rate becomes larger than some value.

Interestingly, the evolution of change in the total number of infections mid-river as the river flow increases is quite dependent on initial conditions.
First, note that in our experiments (not shown), we always saw a decrease and increase of the total number of infections upriver and downriver, respectively.
Now concerning the mid-river locations, if instead of assuming $I_i(0)=0$, $W_1(0)=500$ and $W_2(0)=W_3(0)=0$, we use (not shown) $I_i(0)=1$, $W_i(0)=0$, then the three curves would start at the same point when the percentage of base flow rate is $0\%$ and the total number of infections would be decreasing for the mid-river location.
The same decreasing behaviour is also observed for ``balanced'' initial conditions equal in all locations.
It seems that a situation as in Figure~\ref{fig:final-size-high-flow-river}, with an initial increase of the total number of cases mid-river, is associated to larger initial conditions upriver, whether it is only in humans, in the water or in both.

One last remark: remember that at the end of Section~\ref{sec:analysis}, we discussed the fact that in metapopulation models, the system-wide basic reproduction number $\R_0$ is often in the interval $[\min\R_0^{(i)},\max\R_0^{(i)}]$.
Although we are at present unable to prove this, the simulation in Figure~\ref{fig:final-size-high-flow-river} shows that one of the hallmark signatures of this property holds.
Indeed, when producing Figure~\ref{fig:final-size-high-flow-river}, we also computed the value of the system-wide reproduction number $\R_0$ for each value and observed that for any value of the multiplier of water movement used, $\R_0=5$.

\section{Discussion}
\label{sec:discussion}
In the mathematical analysis, some points merit further exploration, most important of them the link between $\R_0^{(i)}$ and $\R_0$ as well as the endemic situation.
We suspect that these will follow the usual pattern observed in this type of epidemiological model: 1) the system tends to a unique endemic equilibrium when $\R_0>1$ and 2) the usual bounding of $\R_0$ between the minimum and maximum values of $\R_0^{(i)}$ in isolation holds.
This is what we observe numerically in the large number of simulations carried out and although this does not offer a proof, this together with our experience with this type of models means that we are fairly confident that such results could be conjectured.

What is more surprising is how little influence water movement has in the scenarios we considered. Of course, we considered a limited number of locations, but in practice, all of the scenarios considered essentially had the same behaviour.
The flow of water had no discernable influence as soon as infected humans were able to move between locations.
This suggests that incorporating water movement, either explicitly as we do or implicitly as is done in \cite{tuite2011cholera}, is not a requirement of such models: for most use cases, it probably suffices to consider water as a compartment specific to each location and not subject to implicit or explicit influence by other locations.

One specific scenario might be interesting to explore further regarding the effect of water movement is that of the initial spread of the infection, when the number of humans contaminated is very small.
However, the beginning of an epidemic is typically badly captured by deterministic models and stochastic models such as continuous Markov chains are much better suited.
However, running a continuous time Markov chain is difficult in the case of an infection like cholera, since bacteria and humans operate on completely different scales. 
Human populations number in hundred, thousands, up to millions in the models under consideration.
Bacteria, on the other hand, are present in tremendously large numbers; for instance, \cite{larocque2008proteomic} found a median of $5.4\times 10^6$ (range, $7.0\times 10^2$ to $1.2\times 10^8$) colony forming units (CFU) of \emph{V. cholerae} per \emph{millilitre} of stool sample collected from each of 32 cholera patients upon presentation to a hospital in Dhaka, Bangladesh.

Our model does not incorporate a latent compartment for humans as the latent period for cholera is relatively short, from 1 to 5 days, with about 5\% of cases occurring after only 0.5 days \cite{azman2013incubation}.
If we were to include such a compartment, the model could be made more realistic by assuming then that movement of infectious individuals is unlikely.
However, the numerical considerations here suggest that even if we were to integrate such a process, because we are considering short length movement, 
the behaviour would be the same as observed here: it suffices for humans carrying the pathogen to be able to travel for it to spread widely.

\subsubsection*{Acknowledgements}
JA acknowledges support from NSERC's Discovery Grants program.
PMTD acknowledges support from CNRS-Africa -- Joint Research Program ``FANE-MATH-PE''.

\subsubsection*{Data availability statement}
The \texttt{R} code used in the computational analysis will be made available on a public github repository.

\appendix

\bibliographystyle{plain}
\bibliography{biblio,biblio_Arino_Julien,epidemio,math_epi,math}
\end{document}